\documentclass[english]{article}
\usepackage[margin=2.5cm]{geometry}
\usepackage[hidelinks]{hyperref}
\usepackage{amsthm,amsfonts,amssymb,amsmath}
\usepackage{graphicx}

\newtheorem{theorem}{Theorem}
\newtheorem{proposition}[theorem]{Proposition}
\newtheorem{lemma}[theorem]{Lemma}

\newtheorem{corollary}[theorem]{Corollary}

\theoremstyle{definition}

\numberwithin{equation}{section}
\numberwithin{figure}{section}

\usepackage{babel}
\usepackage[T1]{fontenc}
\usepackage[utf8]{inputenc}

\newcommand{\su}{\subseteq}

\renewcommand{\S}{\mathcal{S}}
\newcommand*\diff{\mathop{}\!\mathrm{d}}

\newcommand{\PP}{\operatorname{\mathbb{P}}}

\begin{document}
\title{On the probability of a Condorcet winner among a large number of alternatives}
\author{Lisa Sauermann\thanks{Massachusetts Institute of Technology, MA 02139. Email: {\tt lsauerma@mit.edu}. Research supported by NSF Award DMS-2100157.}}

\maketitle

\begin{abstract}\noindent
Consider $2k-1$ voters, each of which has a preference ranking between $n$ given alternatives. An alternative $A$ is called a Condorcet winner, if it wins against every other alternative $B$ in majority voting (meaning that for every other alternative $B$ there are at least $k$ voters who prefer $A$ over $B$). The notion of Condorcet winners has been studied intensively for many decades, yet some basic questions remain open. In this paper, we consider a model where each voter chooses their ranking randomly according to some probability distribution among all rankings. One may then ask about the probability to have a Condorcet winner with these randomly chosen rankings (which, of course, depends on $n$ and $k$, and the underlying probability distribution on the set of rankings). In the case of the uniform probability distribution over all rankings, which has received a lot of attention and is often referred to as the setting of an ``impartial culture'', we asymptotically determine the probability of having a Condorcet winner for a fixed number $2k-1$ of voters and $n$ alternatives with $n\to \infty$. This question has been open for around fifty years. While some authors suggested that the impartial culture should exhibit the lowest possible probability of having a Condorcet winner, in fact the probability can be much smaller for other distributions. We determine, for all values of $n$ and $k$, the smallest possible probability of having a Condorcet winner (and give an example of a probability distribution over all rankings which achieves this minimum possible probability).
\end{abstract}

\section{Introduction}

For some fixed number $k\geq 1$, suppose that $2k-1$ friends are planning to go on a vacation trip, staying all together in a vacation rental. They are looking at a large number of options, say $n$ options, that are listed on major vacation rental websites. In order for them to agree on booking a particular vacation rental $A$, it would be desirable that there is no other vacation rental $B$ which a majority of the friends prefers over $A$ (as otherwise it would make more sense to rather book option $B$ instead of $A$ if the majority of the group prefers that). In other words, for the group to decide to book some vacation rental $A$, there should be no vacation rental $B$ such that at least $k$ out of the $2k-1$ friends like $B$ better than $A$. What is the probability that the group of friends can indeed find a vacation rental $A$ with this property?

More formally, let $\S$ be the set of the $n$ vacation rental options, i.e.\ the set of the $n$ alternatives that the group is considering.  Each of the $2k-1$ friends can be considered to be a ``voter'' that has a particular preference ranking of the alternatives in $\S$. Let $P_{\S}$ be the set of all $n!$ possible rankings of $\S$ (i.e.\ the set of permutations of $\S$), and let $\sigma_1,\dots,\sigma_{2k-1}\in P_{\S}$ be the rankings of the $2k-1$ voters. 

We remark that in this paper, we always assume the set $\S$ of alternatives to be finite, albeit potentially very large. The setting where the set of alternatives is an infinite set of points forming continuous topological space has been studied for example by Plott \cite{plott}, McKelvey \cite{mckelvey-1,mckelvey-2}, and Schofield \cite{schofield}.

An alternative $A\in \S$ is called a \emph{Condorcet winner} if for every other alternative $B\in \S$ there are at least $k$ indices $i\in \{1,\dots,2k-1\}$ such that $A$ is ranked higher than $B$ in $\sigma_i$ (meaning that at least $k$ of the $2k-1$ voters prefer $X$ over $Y$). In other words, $A$ is a Condorcet winner if it wins against every other alternative $B$ in majority voting. It is easy to see that there can be at most one Condorcet winner in $\S$ (indeed, two different alternatives $A,B\in \S$ cannot both be Condorcet winners, since only of $A$ and $B$ wins in majority voting between $A$ and $B$). The notion of a Condorcet winner is named after Nicolas de Condorcet, who noted in 1785 \cite{condorcet} that already for $n=3$ alternatives it can happen that no such Condorcet winner exists. The observation that there does not always exist a Condorcet winner among the given alternatives is maybe somewhat suprising, and is commonly referred to as \emph{Condorcet's paradox}. There is an extensive body of research about Condorcet winners, see Gehrlein's book on this topic \cite{gehrlein-book} for an overview.

Often, the Condorcet winner problem is considered in the setting of elections, where typically there are many voters who vote on a small number of candidates (see for example \cite{bell, guilbaud, mossel, niemi-weisberg}). However, as in our example above concerning the group of friends looking for a vacation rental, there are also many settings where one has a small number of voters and many alternatives. In this paper, we will be particularly interested in the setting of having  $2k-1$ voters for fixed $k\geq 1$ and a large number $n$ of alternatives.

Our question above concerning the probability of the group of friends finding a vacation rental can now be rephrased as asking about  the probability to have a Condorcet winner in $\S$ with respect to (random) rankings $\sigma_1,\dots,\sigma_{2k-1}\in P_{\S}$. Of course, this depends on the random model for choosing $\sigma_1,\dots,\sigma_{2k-1}\in P_{\S}$.

A simple and natural model is to choose $\sigma_1,\dots,\sigma_{2k-1}\in P_{\S}$ independently and uniformly at random among all rankings in $P_\S$ (i.e.\ among all possible $n!$ rankings of the $n$ alternatives). In our example above, this would mean that there is no bias between the different vacation rental options, and so for each of the friends any of the $n!$ possible rankings of the $n$ options is equally likely (and the rankings of the $2k-1$ friends are all independent). In the literature, this is called an \emph{impartial culture}, and the problem of determining the probability of having a Condorcet winner in this setting has been studied since 1968. Specifically, Garman and Kamien \cite{garman-kamien} as well as DeMeyer and Plott \cite{demeyer-plott} calculated the probability of having a Condorcet wnner for various small values of $n$ and $k$ and conjectured that this probability tends to zero as $n\to \infty$ for a fixed number of $2k-1$ voters (further calculations  for small values of $n$ and $k$ can be found in \cite{gehrlein-fishburn}). This conjecture was proved by May \cite{may} who also asymptotically determined the probability of a Condorcet winner for $3$ voters (i.e.\ for $k=2$) and a large number of alternatives. However, for $k>2$ even the order of magnitude of the probability of having a Condorcet winner was not known. Our first result asymptotically determines this probability.

\begin{theorem}\label{thm-impartial-culture}
For any fixed number $k\geq 1$ and a (large) set $\S$ of $n$ alternatives, let us consider $2k-1$ voters with independently and uniformly random rankings $\sigma_1,\dots,\sigma_{2k-1}\in P_{\S}$. Then the probability that there is Condorcet winner is
\[C_k\cdot n^{-(k-1)/k}+O_k\left(\frac{(\ln n)^{1/k}}{n}\right),\]
where $0<C_k<\infty$ is given by the $(2k-1)$-fold integral
\begin{equation}\label{eq-def-C-k-integral}
C_k=\int_0^\infty \dots \int_0^\infty \exp(-\sigma_k(x_1,\dots,x_{2k-1}))\diff x_1 \dots \diff x_{2k-1}.
\end{equation}
\end{theorem}

Here, $\sigma_k(x_1,\dots,x_{2k-1})=\sum_{I\su \{1,\dots,2k-1\}, |I|=k} \prod_{i\in I} x_i$ denotes the usual $k$-th elementary symmetric polynomial, i.e.\ the sum of the products of any $k$ out of the $2k-1$ variables $x_1,\dots,x_{2k-1}$.

In particular, this means that for a fixed number of $2k-1$ voters and a large number of alternatives, in an impartial culture the probability that a Condorcet winner exists is of the form $(C_k+o(1))\cdot n^{-(k-1)/k}$. We remark that for $k=2$ (i.e.\ for $3$ voters) this was already proved by May \cite{may} with $C_3=\pi^{3/2}/2\approx 2.78$, but our error term $O_k((\ln n)^{1/k}/n)$ in Theorem \ref{thm-impartial-culture} is stronger than in May's result. For $k>2$, it was previously not even known that the probability is on the order of $n^{-(k-1)/k}$.

In real life settings, it is often not the case that all possible rankings of the $n$ alternatives are equally likely. In our example above of the friends that are considering options for their vacation, there is likely some correlation between how much a person likes different options. For example, if someone likes to go to the beach, they are likely to rank the vacation rental options on beaches fairly high. In contrast, someone who likes hiking in the mountains is likely to rank the options in the mountains higher. More generally, vacation rental options that are similar in some way (e.g.\ being on the beach or being in the mountains) are more likely to be close to each other in someone's ranking. It therefore makes sense to also consider other, non-uniform probability distributions on the set $P_\S$ of all rankings of the $n$ alternatives.

So let us consider a probability distribution $\pi$ on the set $P_\S$ of all $n!$ rankings of $\S$. This probability distribution corresponds to how likely it is for a person to have a specific ranking of the $n$ alternatives in $\S$. In the literature, such a probability distribution is often referred to as a \emph{culture}. In an impartial culture, as in Theorem \ref{thm-impartial-culture} the probability distribution $\pi$ is uniform on $P_\S$ (i.e.\ it assigns probability $1/n!$ to each ranking in $P_\S$), but now we allow any probability distribution (which can take into account that it is likely to for similar alternatives to be ranked close to each other). Note that it is in particular possible for $\pi$ to assign probability zero to some rankings in $P_\S$ (if there are rankings that cannot reasonably occur).

Now, consider rankings $\sigma_1,\dots,\sigma_{2k-1}$ that are chosen independently according to the probability distribution $\pi$ on $P_\S$. One may again ask what the probability of having a Condorcet winner is in this setting. Obviously, this depends on the probability distribution $\pi$.

It is easy to see that the probability of a Condorcet winner can be as large as $1$ (i.e.\ there may always exist a Condorcet winner). For example, if $\pi$ is a probability distribution that assigns a particular ranking of $\S$ probability $1$ (and all other rankings probability $0$), then $\sigma_1,\dots,\sigma_{2k-1}$ are always equal to this particular ranking and so the highest-ranked alternative in this ranking is clearly a Condorcet winner. In other words, if the probability distribution $\pi$ is concentrated on just a single ranking in $P_\S$, there is always a Condorcet winner.

One might expect that the lowest possible probability of having a Condorcet winner occurs in the opposite case, where $\pi$ is uniformly distributed among all rankings in $P_\S$ (i.e.\ in the setting of an impartial culture, as in Theorem \ref{thm-impartial-culture}). However, perhaps surprisingly, this is not true. The probability of a Condorcet winner can in fact be much smaller than in Theorem \ref{thm-impartial-culture}. Specifically, consider the probability distribution $\pi^*$ defined as follows. Let $\S=\{A_1,\dots,A_n\}$ and consider the $n$ ``cyclic-looking'' rankings of the form $(A_i, A_{i+1},....,A_n,A_1,\dots,A_{i-1})$ for $i=1,\dots,n$. Now let $\pi^*$ assign probability $1/n$ to each of these $n$ rankings (and probability $0$ to all other rankings). Then it turns out (see below) that for a large number $n$ of alternatives, the probability of a Condorcet winner is only on the order of $n^{-(k-1)}$ (which is much smaller than $n^{-(k-1)/k}$ as in Theorem \ref{thm-impartial-culture}).

A similar phenomenon was observed by Garman and Kamien \cite[p.~315]{garman-kamien} in the setting of $n=3$ alternatives and a large number of voters. They \cite[p.~314]{garman-kamien} suggested to classify a probability distribution $\pi$ on $P_\S$ as ``similar'' if (for a large number of voters) the probability of having a Condorcet winner is bigger than for the uniform distribution (i.e.\ for the impartial culture), and to classify $\pi$ as ``antagonistic'' if the probability of having a Condorcet winner is smaller than for the uniform distribution. However, while it is clear that that the maximum possible probability of having a Condorcet winner is $1$, the actual minimum possible probability has not been determined.

Answering this question, our next theorem determines the minimum  possible probability of having a Condorcet winner for $2k-1$ voters and $n$ alternatives for any positive integers $n$ and $k$. This minimum possible probability of having a Condorcet winner is attained by the probability distribution $\pi^*$ described above. In other words, in the setting of a large number of alternatives discussed above, this probability distribution $\pi^*$ does not only yield to a much smaller probability for a Condorcet winner than the uniform distribution, but $\pi^*$ actually gives the smallest possible probability among all probability distributions on $P_\S$.

\begin{theorem}\label{thm-minimum-probability}
For some $n\geq 1$, let $\S=\{A_1,\dots,A_n\}$ be a set of $n$ alternatives. Then for any $k\geq 1$ and any probability distribution $\pi$ on the set $P_\S$ of all rankings of $\S$, the probability that there is a Condorcet winner when $2k-1$ voters independently choose rankings $\sigma_1,\dots,\sigma_{2k-1}\in P_\S$ according to the probability distribution $\pi$ is at least
\begin{equation}\label{eq-term-minimum-probability}
n^{-(2k-2)}\cdot \sum_{\ell=0}^{k-1}\binom{2k-1}{\ell}(n-1)^\ell.
\end{equation}
Furthermore, the probability of having a Condorcet winner is equal to the term in (\ref{eq-term-minimum-probability}) for the probability distribution $\pi^*$ on $P_\S$ defined by taking each of the rankings $(A_i, A_{i+1},....,A_n,A_1,\dots,A_{i-1})$ for $i=1,\dots,n$ with probability $1/n$, and all other rankings with probability $0$.
\end{theorem}

We stress again that Theorem \ref{thm-minimum-probability} applies to any positive numbers $n$ and $k$, and does not require an assumption of the number $n$ of alternatives being large compared to $k$. Buckley \cite[p.\ 113]{buckley} proved the special case of $n=3$ and $k=2$ of Theorem  \ref{thm-minimum-probability} in 1975 (i.e.\ in the case of three voters deciding between three alternatives).

The term  (\ref{eq-term-minimum-probability}) that gives the precise answer for the minimum possible probability of having a Condorcet winner cannot be simplified for general $n$ and $k$ (there is unfortunately no way to write this sum in a closed form). However, there are ways to express this term asymptotically if $n$ is large with respect to $k$ or vice versa.

If, as in the setting discussed earlier, the number $2k-1$ of voters is fixed, and the number of alternatives is larger (as in the example of the $2k-1$ friends looking for a vacation rental), then the minimum possible probability as determined by Theorem \ref{thm-minimum-probability} has the form
\[n^{-(2k-2)}\cdot \binom{2k-1}{k-1}(n-1)^{k-1}+O_k(n^{-k})=\binom{2k-1}{k}\cdot n^{-(k-1)}+O_k(n^{-k}).\]
We remark that for large $n$ this probability is much smaller than the probability of having a Condorcet winner for the uniform probability distribution on $P_\S$ (i.e.\ for an impartial culture). Indeed, this the latter probability is asymptotically equal to $C_k \cdot n^{-(k-1)/k}$ by Theorem \ref{thm-impartial-culture}.

On the other hand, if there is a fixed number $n\geq 3$ of alternatives and there are $2k-1$ voters for $k\to\infty$, then the minimum possible probability as determined by Theorem \ref{thm-minimum-probability} has the form
\[n^{-2k}\cdot \binom{2k-1}{k-1}\cdot (n-1)^{k}\cdot \exp(O_n(1))=\exp\left(-\ln\left(\frac{n^2}{4(n-1)}\right)\cdot k+O_n(\log k)\right).\]
So for fixed $n\geq 3$, this probability decays exponentially with $k$. Again, for the uniform probability distribution (i.e.\ for an impartial culture) the probability of having a Condorcet winner is much larger. In fact, for every fixed number $n\geq 3$ of alternatives, the latter probability converges to some real number strictly between $0$ and $1$ as $k\to\infty$ (see e.g.\ \cite[p. 320]{niemi-weisberg}). Note that for $n\in \{1,2\}$ alternatives, the probability of having a Condorcet winner is always equal to $1$ for any probability distribution on $P_\S$.

These results are in sharp contrast to Gehrlein's claim \cite{gehrlein-2002} that the probability of having a Condorcet winner is minimal for an impartial culture (i.e.\ for the uniform distribution $\pi$). More precisely, Gehrlein \cite[p.\ 197]{gehrlein-2002} argues that an impartial culture (as well as other models of ``balanced preferences'' that are described in his paper) gives a lower bound for the probability of a Condorcet winner in more general situations. He \cite[p.\ 177]{gehrlein-2002} also writes ``In general, intuition suggests that we would be most likely to observe this paradox on the pairwise majority rule relations when voters' preferences are closest to being balanced between all pairs of candidates'' and, with regards to this statement, ``it does seem to be a generally valid claim''. However, as shown above, for the probability distribution $\pi^*$ defined in Theorem \ref{thm-minimum-probability}, a Condorcet winner exists with much smaller probability than in an impartial culture (i.e.\ for the uniform probability distribution on $P_\S$). This disproves Gehrlein's claim.

In fact, this probability distribution $\pi^*$ achieves the actual minimum possible probability of having a Condorcet winner. Interestingly, while the distribution $\pi^*$ is balanced (or more precisely, symmetric) between all of the candidates, it is not ``balanced between all pairs of candidates'' as Gehrlein \cite[p.\ 177]{gehrlein-2002} suggested. For example, looking at alternatives $A_1$ and $A_2$, a voter whose ranking $\sigma$ is chosen according to the probability distribution $\pi^*$ prefers $A_1$ over $A_2$ with probability $(n-1)/n$. Thus, for $n\geq 2$ the probability distribution $\pi^*$ is far from balanced when looking at the pair of alternatives $A_1$ and $A_2$. Still, as shown in Theorem \ref{thm-minimum-probability}, the probability distribution $\pi^*$ minimizes the probability of a Condorcet winner among all possible probability distributions on $P_\S$. This again strongly disproves Gehrlein's claim \cite[p.\ 177]{gehrlein-2002}.

We remark that Tsetlin et al.\ \cite{tsetlin-et-al} studied the problem of minimizing the probability of a Condorcet winner in the setting of three candidates under the additional assumption that the probability distribution $\pi$ induces a transitive weak majority preference relationship (see the assumption of \cite[Theorem 3]{tsetlin-et-al}). Under this additional restriction, they proved that the impartial culture minimizes the probability of a Condorcet winner for three alternatives, and they conjectured the same to be true for more than three alternatives.

In conclusion, depending on the probability distribution $\pi$, the probability of having a Condorcet winner can be anywhere between $1$ and the probability in (\ref{eq-term-minimum-probability}). If there is no bias or correlation between the $n$ alternatives, i.e.\ in the setting of an impartial culture, then the probability is asymptotically equal to $C_k\cdot n^{-(k-1)/k}$ with $C_k$ as in Theorem \ref{thm-impartial-culture}. This probability is in some sense neither close to the upper bound $1$ nor to the lower bound  in (\ref{eq-term-minimum-probability}). Coming back to our example from the beginning, one should hope that the relevant probability distribution $\pi$ leads to a much higher probability of a Condorcet winner, since it would be frustrating for the group of friends to be able to agree on a vacation rental option only with probability tending to zero for large $n$. It is indeed plausible to expect that, given that most likely some of the vacation rental options are inherently better than others, and so one should expect the probability distribution $\pi$ to be fairly biased towards these ``better'' alternatives. In this light, it is not surprising that in practice a group of friends is usually able to find a vacation rental that suits their needs.

\textit{Notation and Organization.} For a positive integer $m$, we abbreviate the set $\{1,\dots,m\}$ by $[m]$ (which is a common notation). We use standard asymptotic $O$-notation for fixed $k\geq 1$ and $n\to \infty$. In order to emphasize the possible dependence on the fixed number $k$, we write, for example, $O_k(1/n)$ for a term which is bounded in absolute value by a term of the form $D_k \cdot (1/n)$ for some constant $D_k>0$ depending on $k$.

For $\ell\leq m$, and variables $x_1,\dots,x_m\in \mathbb{R}$, the $\ell$-th elementary symmetric polynomial of $x_1,\dots,x_m$ is defined to be
\[\sigma_\ell(x_1,\dots,x_{m})=\sum_{\substack{I\su [m]\\ |I|=\ell}} \,\prod_{i\in I} x_i,\]
i.e.\ it is the sum of the products of any $\ell$ out of the $m$ variables $x_1,\dots,x_{m}$. Note that this is a homogeneous polynomial of degree $\ell$, and, furthermore, for non-negative $x_1,\dots,x_m$, this polynomial $\sigma_\ell(x_1,\dots,x_{m})$ is always non-negative.

We will prove Theorem \ref{thm-impartial-culture} in Section \ref{sect-proof-thm-impartial-culture}, postponing the proofs of some lemmas to Section \ref{sect-lemmas}. Afterwards, we prove Theorem \ref{thm-minimum-probability} in Section \ref{sect-proof-minimum-probability}.

\textit{Acknowledgements.} The author would like to thank Asaf Ferber, Matthew Kwan and Elchanan Mossel for helpful discussions.

\section{The probability of a Condorcet winner in an impartial culture}
\label{sect-proof-thm-impartial-culture}

In this section, we prove Theorem \ref{thm-impartial-culture}, which concerns the probability of having a Condorcet winner in an impartial culture for $2k-1$ voters and a large number of alternatives.

First, note that the theorem is trivially true for $k=1$. Indeed, for $k=1$, there is only one voter and therefore a Condorcet winner exists with probability $1$. On the other hand,
\[C_1=\int_0^\infty \exp(-\sigma_1(x_1))\diff x_1 =\int_0^\infty \exp(-x_1)\diff x_1 =1\]
and therefore we have $C_1\cdot n^{(1-1)/1}=1$, as desired.

So let us from now on fix $k\geq 2$. Since we are proving an asymptotic statement, we may assume that $n$ is sufficiently large (with respect to $k$). In particular, we can assume that $2(\ln n)/(n-1)\leq 1/3$ .

We have a set $\S=\{A_1,\dots,A_n\}$ of $n$ alternatives and $2k-1$ voters choose rankings $\sigma_1,\dots,\sigma_{2k-1}$ independently and uniformly at random form the set $P_\S$ of all $n!$ rankings of $\S$. We can model the random choice of these rankings as follows.

Let us assume that each voter $i$ for $i=1,\dots,2k-1$ picks $n$ random real numbers $x_i^{(1)},\dots,x_i^{(n)}$ independently uniformly at random from the interval $[0,1]$ (and this happens independently for all voters). Then with probability $1$ these $n$ points $x_i^{(1)},\dots,x_i^{(n)}$ are distinct. Let the ranking $\sigma_i$ of $\S=\{A_1,\dots,A_n\}$ be obtained by recording the order of the points $x_i^{(1)},\dots,x_i^{(n)}$ in the interval $[0,1]$, meaning that voter $i$ ranks the alternative $A_\ell$ first (i.e.\ highest) for which the number $x_i^{(\ell)}$ is the smallest among $x_i^{(1)},\dots,x_i^{(n)}$ (and then the alternative $A_{\ell'}$ second for which the $x_i^{(\ell')}$ is the second-smallest number and so on). This way, we indeed obtain a uniformly random ranking $\sigma_i\in P_\S$ for voter $i$, and these rankings are independent for all the voters $i=1,\dots,2k-1$.

Note that since there is always at most one Condorcet winner, the total probability that a Condorcet winner exists is the sum of of the probabilities that $A_\ell$ is a Condorcet winner for all $\ell=1,\dots,n$. In other words,
\[\PP(\text{Condorcet winner exists})=\sum_{\ell=1}^{n}\PP(A_\ell\text{ is  Condorcet winner}).\]
Note that since the rankings $\sigma_i$ are independent uniformly random in $P_\S$, each alternative $A_\ell$ for $\ell=1,\dots,n$ is equally likely to be a Condorcet winner. Hence the summands on the right-hand side of the previous equation are equal for all $\ell=1,\dots,n$ and we obtain that
\[\PP(\text{Condorcet winner exists})=\sum_{\ell=1}^{n}\PP(A_\ell\text{ is  Condorcet winner})=n\cdot \PP(A_1\text{ is  Condorcet winner}).\]
Hence, in order to prove Theorem \ref{thm-impartial-culture}, it suffices to show that
\begin{equation}\label{eq-term-prob-one-alternative}
\PP(A_1\text{ is  Condorcet winner}) = C_k\cdot n^{-(2k-1)/k}+O_k\left(\frac{(\ln n)^{1/k}}{n^2}\right),
\end{equation}
where $C_k$ is given by (\ref{eq-def-C-k-integral}) and that the integral in (\ref{eq-def-C-k-integral}) is indeed finite (it is clear that the integral is positive, but it could a priori be infinite). The finiteness of the integral follows from the first part of the following more general lemma.

\begin{lemma}\label{lem-integral-finite} For any positive integers $\ell<m$, we have
\[\int_0^\infty \dots \int_0^\infty \exp(-\sigma_\ell(x_1,\dots,x_{m}))\diff x_1 \dots \diff x_{m}\leq (m!)^2\]
Furthermore, defining $0<C_{\ell,m}<\infty$ to be the value of this integral, then in the case of $\ell\geq 2$ we have
\[\int_0^a \dots \int_0^a \exp(-\sigma_\ell(x_1,\dots,x_{m}))\diff x_1 \dots \diff x_{m}\geq C_{\ell,m} - m\cdot (m-1) \cdot ((m-1)!)^2 \cdot a^{-(m-\ell)/(\ell-1)}\]
for every $a\geq 1$.
\end{lemma}

We postpone the proof of Lemma \ref{lem-integral-finite} to Section \ref{sect-lemmas}. Applying Lemma \ref{lem-integral-finite} to $\ell=k$ and $m=2k-1$ (noting that $k<2k-1$ by our assumption that $k\geq 2$) implies that the integral in (\ref{eq-def-C-k-integral}) is at most $((2k-1)!)^2$, and so in particular it is finite. So we have $0<C_k<\infty$ and it now suffices to prove (\ref{eq-term-prob-one-alternative}).

Recall that alternative $A_1$ is a Condorcet winner if and only if for each $\ell=2,\dots,n$ there are at least $k$ voters $i\in [2k-1]$ such that voter $i$ ranks $A_1$ before $A_\ell$. In other words, for each $\ell=2,\dots,n$ there must be at least $k$ different indices $i\in [2k-1]$ such that $x_i^{(1)}\leq x_i^{(\ell)}$.

Remembering that all $x_i^{(\ell)}$ for $i=1,\dots,2k-1$ and $\ell=1,\dots, n$ are independent uniformly random real numbers in the interval $[0,1]$, we can imagine that we first sample the random numbers $x_1^{(1)},\dots,x_{2k-1}^{(1)}$ (i.e.\ the numbers corresponding to alternative $A_1$). For simplicity, let us write $x_i=x_i^{(1)}$ for $i=1,\dots,2k-1$. Then $A_1$ is a Condorcet winner if and only if  for each $\ell=2,\dots,n$ the random numbers $x_1^{(\ell)},\dots,x_{2k-1}^{(\ell)}$ satisfy the condition
\begin{equation}\label{eq-condition-for-each-ell}
x_i\leq x_i^{(\ell)}\text{ for at least }k\text{ indices }i\in [2k-1].
\end{equation}

Given the values of $x_1,\dots,x_{2k-1}\in [0,1]$, let $Q(x_1,\dots,x_{2k-1})$ be the probability that for independent uniformly random variables $y_1,\dots,y_{2k-1}$ in the interval $[0,1]$ we have $x_i\leq y_i$ for at most $k-1$ indices $i\in [2k-1]$. Then each $\ell=2,\dots,n$ satisfies  condition (\ref{eq-condition-for-each-ell}) with probability precisely $1-Q(x_1,\dots,x_{2k-1})$ and these events are independent for all $\ell=2,\dots,n$ (conditioning on the values of $x_1,\dots,x_{2k-1}$). Hence, conditioning on the values of $x_1,\dots,x_{2k-1}$, the probability of $A_1$ being Condorcet winner is $(1-Q(x_1,\dots,x_{2k-1}))^{n-1}$. Thus, we obtain
\begin{equation}\label{eq-expression-with-Q}
\PP(A_1\text{ is  Condorcet winner}) = \int_0^1 \dots \int_0^1 \left(1-Q(x_1,\dots,x_{2k-1})\right)^{n-1}\diff x_1 \dots \diff x_{2k-1}.
\end{equation}

The following lemma gives some helpful estimates for $Q(x_1,\dots,x_{2k-1})$ in terms of $\sigma_k(x_1,\dots,x_{2k-1})$.

\begin{lemma}\label{lemma-bounds-Q} For real numbers $x_1,\dots,x_{2k-1}\in [0,1]$, we have
\[2^{-2k+1}\cdot \sigma_k(x_1,\dots,x_{2k-1})\leq Q(x_1,\dots,x_{2k-1})\leq \sigma_k(x_1,\dots,x_{2k-1})\]
and furthermore
\[Q(x_1,\dots,x_{2k-1})\geq \sigma_k(x_1,\dots,x_{2k-1})-2^{4k-2}\cdot (\sigma_k(x_1,\dots,x_{2k-1}))^{(k+1)/k}.\]
\end{lemma}

We postpone the proof of Lemma \ref{lemma-bounds-Q} to Section \ref{sect-lemmas}. The proof relies on Bonferroni's inequalities in probability theory, as well as on Newton's inequality for elementary symmetric functions. For our argument, we will use the following corollary of Lemma \ref{lemma-bounds-Q}.

\begin{corollary}\label{coro-Q-close-sigma-k}
If $x_1,\dots,x_{2k-1}\in [0,1]$ are real numbers with $Q(x_1,\dots,x_{2k-1})\leq 2 (\ln n)/(n-1)$, then we have
\[Q(x_1,\dots,x_{2k-1})\geq \left(1-2^{4k}\cdot \left(\frac{\ln n}{n-1}\right)^{1/k}\right)\cdot \sigma_k(x_1,\dots,x_{2k-1}).\]
\end{corollary}

Recall that by Lemma \ref{lemma-bounds-Q}, we always have  $Q(x_1,\dots,x_{2k-1})\leq \sigma_k(x_1,\dots,x_{2k-1})$ if $x_1,\dots,x_{2k-1}\in [0,1]$. This means that under the assumptions in Corollary \ref{coro-Q-close-sigma-k}, the value of $Q(x_1,\dots,x_k)$ is actually fairly close to $\sigma_k(x_1,\dots,x_{2k-1})$.

\begin{proof}[Proof of Corollary \ref{coro-Q-close-sigma-k} assuming Lemma \ref{lemma-bounds-Q}]
When combining the assumption $Q(x_1,\dots,x_{2k-1})\leq 2 (\ln n)/n$ of the corollary with the first inequality in Lemma \ref{lemma-bounds-Q}, we obtain
\[\sigma_k(x_1,\dots,x_{2k-1})\leq 2^{2k-1}\cdot Q(x_1,\dots,x_{2k-1})\leq 2^{2k-1}\cdot 2 \cdot \frac{\ln n}{n-1} =4^k\cdot \frac{\ln n}{n-1}.\]
Hence
\[\sigma_k(x_1,\dots,x_{2k-1})^{(k+1)/k}=\sigma_k(x_1,\dots,x_{2k-1})^{1/k}\cdot \sigma_k(x_1,\dots,x_{2k-1})\leq  4\cdot \left(\frac{\ln n}{n-1}\right)^{1/k}\cdot  \sigma_k(x_1,\dots,x_{2k-1}),\]
and from the second part of Lemma \ref{lemma-bounds-Q} we obtain
\[Q(x_1,\dots,x_{2k-1})\geq \sigma_k(x_1,\dots,x_{2k-1})-2^{4k-2} \sigma_k(x_1,\dots,x_{2k-1})^{(k+1)/k} \geq \left(1-2^{4k} \left(\frac{\ln n}{n-1}\right)^{1/k}\right) \sigma_k(x_1,\dots,x_{2k-1}),\]
as desired.
\end{proof}

In order to estimate the integral in (\ref{eq-expression-with-Q}), we also need the following lemma, which states, roughly speaking, that for $0\leq t\leq 1/3$, the function $1-t$ is fairly close to $e^{-t}$. The lemma follows relatively easily from Taylor's theorem, and we postpone the proof details to Section \ref{sect-lemmas}.

\begin{lemma}\label{lem-taylor}
For $0\leq t\leq 1/3$, we have
\[e^{-t-t^2}\leq 1-t\leq e^{-t}.\]
\end{lemma}

In order to show (\ref{eq-term-prob-one-alternative}), let us start by rewriting (\ref{eq-expression-with-Q}) by splitting up the $(2k-1)$-fold integral on the right-hand side into two domains. Let
\[D=\{(x_1,\dots,x_{2k-1})\in [0,1]^{2k-1}\mid Q(x_1,\dots,x_{2k-1})\leq 2 (\ln n)/(n-1)\}\]
be the domain of those $(x_1,\dots,x_{2k-1})$ in $[0,1]^{2k-1}$ for which we have $Q(x_1,\dots,x_{2k-1})\leq 2 (\ln n)/(n-1)$. Note that we can apply Corollary \ref{coro-Q-close-sigma-k} to any $(x_1,\dots,x_{2k-1})\in D$. Furthermore, note that for any $(x_1,\dots,x_{2k-1})\in [0,1]^{2k-1}\setminus D$, by Lemma \ref{lem-taylor} we have
\begin{equation}\label{eq-outside-D}
\left(1-Q(x_1,\dots,x_{2k-1})\right)^{n-1}\leq \left(1-2\cdot \frac{\ln n}{n-1}\right)^{n-1}\leq \left(\exp\left(-\frac{2\ln n}{n-1}\right)\right)^{n-1}=n^{-2}.
\end{equation}

From (\ref{eq-expression-with-Q}), we now obtain
\begin{align}
\PP(A_1\text{ is  Condorcet winner}) &= \int_0^1 \dots \int_0^1 \left(1-Q(x_1,\dots,x_{2k-1})\right)^{n-1}\diff x_1 \dots \diff x_{2k-1}\notag\\
&= \int_{[0,1]^{2k-1}} \left(1-Q(x_1,\dots,x_{2k-1})\right)^{n-1} \diff^{2k-1} (x_1,\dots,x_{2k-1})\notag\\
&=\int_{D} \left(1-Q(x_1,\dots,x_{2k-1})\right)^{n-1} \diff^{2k-1} (x_1,\dots,x_{2k-1})\notag\\
&\quad\quad+\int_{[0,1]^{2k-1}\setminus D} \left(1-Q(x_1,\dots,x_{2k-1})\right)^{n-1} \diff^{2k-1} (x_1,\dots,x_{2k-1})\label{eq-sum-two-integrals}
\end{align}

Let us now show upper and lower bounds for sum in (\ref{eq-sum-two-integrals}). First, as an upper bound we have by (\ref{eq-outside-D}) and Lemma \ref{lem-taylor} (noting that $Q(x_1,\dots,x_{2k-1})\leq 2(\ln n)/(n-1)\leq 1/3$ for all $(x_1,\dots,x_{2k-1})\in D$)
\begin{align}
\PP(A_1\text{ is  Condorcet winner})&\leq \int_{D} \left(1-Q(x_1,\dots,x_{2k-1})\right)^{n-1} \diff^{2k-1} (x_1,\dots,x_{2k-1})+\operatorname{Vol}([0,1]^{2k-1}\setminus D)\cdot  n^{-2}\notag\\
&\leq \int_{D} \left(\exp(-Q(x_1,\dots,x_{2k-1}))\right)^{n-1} \diff^{2k-1} (x_1,\dots,x_{2k-1})+ n^{-2}\notag\\\
&= \int_{D} \exp(-Q(x_1,\dots,x_{2k-1})\cdot (n-1)) \diff^{2k-1} (x_1,\dots,x_{2k-1})+ n^{-2}.\label{eq-ineq-step-upper-bound}
\end{align}
Note that for $(x_1,\dots,x_{2k-1})\in D$ we can apply Corollary \ref{coro-Q-close-sigma-k} and obtain
\begin{align*}
Q(x_1,\dots,x_{2k-1})\cdot (n-1)&\geq \left(1-2^{4k}\cdot \left(\frac{\ln n}{n-1}\right)^{1/k}\right)\cdot \sigma_k(x_1,\dots,x_{2k-1})\cdot (n-1)\\
&=\left((n-1)-2^{4k}(\ln n)^{1/k}(n-1)^{(k-1)/k}\right)\cdot \sigma_k(x_1,\dots,x_{2k-1})\\
&\geq \left(n-2^{4k+1}(\ln n)^{1/k}n^{(k-1)/k}\right)\cdot \sigma_k(x_1,\dots,x_{2k-1}).
\end{align*}
Combining this with (\ref{eq-ineq-step-upper-bound}) yields
\begin{align}
&\PP(A_1\text{ is  Condorcet winner})\notag\\
&\quad \leq \int_{D} \exp\left(-\left(n-2^{4k+1}(\ln n)^{1/k}n^{(k-1)/k}\right)\sigma_k(x_1,\dots,x_{2k-1})\right) \diff^{2k-1} (x_1,\dots,x_{2k-1})+ n^{-2}\notag\\
&\quad\leq \int_0^\infty \dots \int_0^\infty \exp\left(-\left(n-2^{4k+1}(\ln n)^{1/k}n^{(k-1)/k}\right)\sigma_k(x_1,\dots,x_{2k-1})\right) \diff x_1 \dots \diff x_{2k-1}+ n^{-2}\notag\\
&\quad= \left(n-2^{4k+1}(\ln n)^{1/k}n^{(k-1)/k}\right)^{-(2k-1)/k}\int_0^\infty \dots \int_0^\infty \exp(-\sigma_k(z_1,\dots,z_{2k-1})) \diff z_1 \dots \diff z_{2k-1}+ n^{-2}\notag\\
&\quad= \left(1+O_k\left(\frac{(\ln n)^{1/k}}{n^{1/k}}\right)\right)\cdot n^{-(2k-1)/k}\cdot C_k+ n^{-2}=C_k\cdot n^{-(2k-1)/k}+O_k\left(\frac{(\ln n)^{1/k}}{n^2}\right),\label{eq-ineq-upper-bound-done}
\end{align}
where the first equality sign is obtained by substituting $z_i=(n-2^{4k+1}k^2(\ln n)^{1/k}n^{(k-1)/k})^{1/k}\cdot x_i$ for $i=1,\dots,2k-1$ (recalling that $\sigma_k$ is a homogeneous polynomial of degree $k$). This finishes the proof of the upper bound in (\ref{eq-term-prob-one-alternative}).

For the lower bound, let us return to (\ref{eq-sum-two-integrals}), and apply Lemma \ref{lem-taylor} to the first integral (recall that $Q(x_1,\dots,x_{2k-1})\leq 2 (\ln n)/(n-1)\leq 1/3$ for all $(x_1,\dots,x_{2k-1})\in D$). This gives
\begin{align}
\PP(A_1\text{ is  Condorcet winner}) &\geq\int_{D} \left(1-Q(x_1,\dots,x_{2k-1})\right)^{n-1} \diff^{2k-1} (x_1,\dots,x_{2k-1})\notag\\
&\geq\int_{D} \exp\left(-(Q(x_1,\dots,x_{2k-1})+Q(x_1,\dots,x_{2k-1})^2)\cdot (n-1)\right) \diff^{2k-1} (x_1,\dots,x_{2k-1})\notag\\
&\geq\int_{D} \exp\left(-Q(x_1,\dots,x_{2k-1})\cdot \left(1+2 \frac{\ln n}{n-1}\right)\cdot (n-1)\right) \diff^{2k-1} (x_1,\dots,x_{2k-1})\notag\\
&\geq\int_{D} \exp\left(-Q(x_1,\dots,x_{2k-1})\cdot (n+2\ln n)\right) \diff^{2k-1} (x_1,\dots,x_{2k-1})\notag\\
&=\int_{D} \exp\left(-Q(x_1,\dots,x_{2k-1})\cdot f(n)\right) \diff^{2k-1} (x_1,\dots,x_{2k-1}),\label{eq-step-loer-bound}
\end{align}
using the short-hand notation $f(n)=n+2\ln n$. The integral in (\ref{eq-step-loer-bound}) is very close to the analogous integral taken over $[0,1]^{2k-1}$ instead of over $D$. Indeed, recall that for any $(x_1,\dots,x_{2k-1})\in [0,1]^{2k-1}\setminus D$ we have $Q(x_1,\dots,x_{2k-1})\cdot f(n)\geq 2 (\ln n)/(n-1)\cdot (n+2\ln n)\geq 2\ln n$ and therefore
\[\int_{[0,1]^{2k-1}\setminus D} \exp\left(-Q(x_1,\dots,x_{2k-1})\cdot f(n)\right) \diff^{2k-1} (x_1,\dots,x_{2k-1})\leq \int_{[0,1]^{2k-1}\setminus D} n^{-2} \diff^{2k-1} (x_1,\dots,x_{2k-1})\leq n^{-2}.\]
Thus, (\ref{eq-step-loer-bound}) implies
\begin{align*}
\PP(A_1\text{ is  Condorcet winner})&\geq\int_{[0,1]^{2k-1}} \exp\left(-Q(x_1,\dots,x_{2k-1})\cdot f(n)\right) \diff^{2k-1} (x_1,\dots,x_{2k-1}) - n^{-2}\\
&\geq\int_0^1 \dots \int_0^1\exp\left(-\sigma_k(x_1,\dots,x_{2k-1})\cdot f(n)\right) \diff x_1\dots \diff x_{2k-1} - n^{-2},
\end{align*}
where for the second inequality we used the upper bound on $Q(x_1,\dots,x_{2k-1})$ in Lemma \ref{lemma-bounds-Q}. Let us now use the substitution $z_i=f(n)^{1/k}\cdot x_i$ for $i=1,\dots,2k-1$ (again recalling that $\sigma_k$ is a homogeneous polynomial of degree $k$). This yields
\begin{align*}
\PP(A_1\text{ is  Condorcet winner})&\geq f(n)^{-(2k-1)/k}\int_0^{f(n)^{1/k}} \dots \int_0^{f(n)^{1/k}}\exp\left(-\sigma_k(x_1,\dots,x_{2k-1})\right) \diff z_1\dots \diff z_{2k-1} - n^{-2},\\
&\geq f(n)^{-(2k-1)/k} \cdot \left(C_k - ((2k)!)^2\cdot f(n)^{-1/k}\right) - n^{-2}\\
&= \left(1-O_k\left(\frac{\ln n}{n}\right)\right)\cdot n^{-(2k-1)/k} \cdot\left(C_k - O_k(n^{-1/k})\right) - n^{-2}\\
&= C_k\cdot n^{-(2k-1)/k}-O_k(n^{-2}),
\end{align*}
where the second inequality follows from the second part of Lemma \ref{lem-integral-finite} applied to $a=f(n)^{1/k}$ with $\ell=k$ and $m=2k-1$ (then $C_{k,2k-1}$ is simply $C_k$ and $(m-\ell)/(\ell-1)=(k-1)/(k-1)=1$), and where in the last step we used that $k\geq 2$. This gives the desired lower bound in (\ref{eq-term-prob-one-alternative}).

\section{The minimum possible probability of a Condorcet winner}
\label{sect-proof-minimum-probability}

In this section, we prove Theorem \ref{thm-minimum-probability}, which determines the minimum possible probability of having a Condorcet winner when $2k-1$ voters choose independent rankings of a given set of $n$ alternatives according to some probability distribution.

For an integer $k\geq 1$ and $x\in \mathbb{R}$, let us define
\[p_k(x)=\sum_{\ell=0}^{k-1}\binom{2k-1}{\ell}\cdot x^{2k-1-\ell}\cdot (1-x)^{\ell}.\]
Note that for $x\in [0,1]$, the value $p_k(x)$ can be interpreted as follows: Consider a biased coin which shows heads with probability $x$ and tails with probability $1-x$. Then $p_k(x)$ is precisely the probability that, when throwing this coin $2k-1$ times, we have at most $k-1$ tails (indeed, each summand in the sum above is the probability of having exactly $\ell$ tails). Equivalently, $p_k(x)$ is the probability that among $2k-1$ throws we have at least $k$ heads.

Also note that we can express the term in (\ref{eq-term-minimum-probability}) in Theorem \ref{thm-minimum-probability} as
\begin{equation}\label{eq-connect-term-to function}
n^{-(2k-2)}\cdot \sum_{\ell=0}^{k-1}\binom{2k-1}{\ell}(n-1)^\ell = n\cdot \sum_{\ell=0}^{k-1}\binom{2k-1}{\ell}\left(\frac{1}{n}\right)^{2k-1-\ell}\left(\frac{n-1}{n}\right)^{\ell}=n\cdot p_k(1/n).
\end{equation}

In other words, in order to prove Theorem \ref{thm-minimum-probability}, we need to show that for every probability distribution $\pi$ as in the theorem statement, the probability of having a Condorcet winner is at least $n\cdot p_k(1/n)$, and that furthermore the probability  is exactly $n\cdot p_k(1/n)$ for the specific probability distribution $\pi^*$ defined in the theorem statement.

Our proof of Theorem \ref{thm-minimum-probability} will crucially rely on the following property of the function $p_k(x)$: For variables $x_1,\dots,x_n\in [0,1]$ consider the optimization problem of minimizing the sum $p_k(x_1)+\dots+p_k(x_n)$ under the constraint that $x_1+\dots+x_n=1$. When taking $x_1=\dots=x_n=1/n$, we obtain $p_k(x_1)+\dots+p_k(x_n)=n\cdot p_k(1/n)$. The following proposition states that this value $n\cdot p_k(1/n)$ is actually optimal, i.e.\ it is the minimum possible value of $p_k(x_1)+\dots+p_k(x_n)$ for any $x_1,\dots,x_n\in [0,1]$ with $x_1+\dots+x_n=1$. The proof of this optimization statement is given later in Section \ref{subsect-propo-optimization} (the proof is somewhat involved and requires several other lemmas that are stated in proved in Section  \ref{subsect-propo-optimization}).

\begin{proposition}\label{propo-optimization}
For every $k\geq 1$ and any real numbers $x_1,\dots,x_n\in [0,1]$ with $x_1+\dots+x_n=1$ we have $p_k(x_1)+\dots+p_k(x_n)\geq n\cdot p_k(1/n)$.
\end{proposition}

We are now ready to prove Theorem \ref{thm-minimum-probability}. We will separately prove the two parts of the theorem (the first part asserting a bound for the probability of having a Condorcet winner for any probability distribution $\pi$, and the second part asserting equality for the specific probability distribution $\pi^*$).

Using (\ref{eq-connect-term-to function}), we can restate the first part of  Theorem \ref{thm-minimum-probability} as follows.

\begin{proposition}\label{prop-minimum-probability-1}
Let $n\geq 1$, let $\S=\{A_1,\dots,A_n\}$ be a set of $n$ alternatives, and let $\pi$ be a probability distribution  on the set $P_\S$ of all rankings of $\S$. Then for any $k\geq 1$, the probability that there is a Condorcet winner when $2k-1$ voters independently choose rankings $\sigma_1,\dots,\sigma_{2k-1}\in P_\S$ according to the probability distribution $\pi$ is at least $n\cdot p_k(1/n)$.
\end{proposition}

\begin{proof}
For each $i=1,\dots,n$, let us define $x_i$ to be the probability that a random ranking of $\mathcal{S}$ chosen according to the probability distribution $\pi$ has $A_i$ as its top-ranked alternative. Then we have $x_1,\dots,x_n\in [0,1]$ and $x_1+\dots+x_n=1$. Hence, by Proposition \ref{propo-optimization} we have $p_k(x_1)+\dots+p_k(x_n)\geq n\cdot p_k(1/n)$.

For any outcomes of the rankings $\sigma_1,\dots,\sigma_{2k-1}\in P_\S$ there is automatically a Condorcet winner if for some $i=1,\dots,n$ at least $k$ of the $2k-1$ rankings $\sigma_1,\dots,\sigma_{2k-1}$ have alternative $A_i$ as their top-ranked alternative (since $A_i$ is automatically a Condorcet winner in this case). We claim that for each $i=1,\dots,n$, this happens with probability exactly $p_k(x_i)$.

Indeed, fix some $i\in \{1,\dots,n\}$. For each of the $2k-1$ random rankings $\sigma_1,\dots,\sigma_{2k-1}$, the probability of $A_i$ being the top-ranked alternative is precisely $x_i$ (by definition of $x_i$). Hence the probability that at least $k$ of the $2k-1$ rankings $\sigma_1,\dots,\sigma_{2k-1}$ have alternative $A_i$ as their top-ranked alternative is the same as the probability that a biased coin showing heads with probability $x_i$ turns up heads at least $k$ times among $2k-1$ throws. This probability is precisely $p_k(x_i)$.

Thus, for each $i=1,\dots,n$ it happens with probability $p_k(x_i)$ that at least $k$ of the $2k-1$ rankings $\sigma_1,\dots,\sigma_{2k-1}$ have alternative $A_i$ as their top-ranked alternative. Furthermore, whenever this happens, alternative $A_i$ is a automatically a Condorcet winner.

Hence, for each $i=1,\dots,n$, alternative $A_i$ is a Condorcet winner with probability at least $p_k(x_i)$. As there is always at most one Condorcet winner, the total probability of having a Condorcet winner is therefore at least
\[p_k(x_1)+\dots+p_k(x_n)\geq n\cdot p_k(1/n),\]
as desired.
\end{proof}

It remains to prove the second part of  Theorem \ref{thm-minimum-probability}. Again using (\ref{eq-connect-term-to function}), we can restate this remaining part as follows.

\begin{proposition}\label{prop-minimum-probability-2}
Let $n\geq 1$, and let $\S=\{A_1,\dots,A_n\}$ be a set of $n$ alternatives. Define $\pi^*$ to be the probability distribution on $P_S$ given by taking each of the rankings $(A_i, A_{i+1},....,A_n,A_1,\dots,A_{i-1})$ for $i=1,\dots,n$ with probability $1/n$, and all other rankings with probability $0$. Then for any $k\geq 1$, the probability that there is a Condorcet winner when $2k-1$ voters independently choose rankings $\sigma_1,\dots,\sigma_{2k-1}\in P_\S$ according to the probability distribution $\pi^*$ is equal to $n\cdot p_k(1/n)$.
\end{proposition}

\begin{proof}
We claim that for any outcomes of the rankings $\sigma_1,\dots,\sigma_{2k-1}$ chosen according to the probability distribution $\pi^*$, alternative $A_j$ is a Condorcet winner if and only if at least $k$ of the $2k-1$ rankings $\sigma_1,\dots,\sigma_{2k-1}$ are equal to $(A_j, A_{i+1},....,A_n,A_1,\dots,A_{j-1})$. Indeed, if at least $k$ of the $2k-1$ rankings are $(A_j, A_{j+1},....,A_n,A_1,\dots,A_{j-1})$, then alternative $A_j$ is the first-ranked alternative for at least $k$ of the $2k-1$ voters and so $A_j$ must be a Condorcet winner. On the other hand, if some alternative $A_j$ is a Condorcet winner for some outcome of the rankings $\sigma_1,\dots,\sigma_{2k-1}$, then at least $k$ of the $2k-1$ voters must rank alternative $A_j$ higher than alternative $A_{j-1}$. However, the only ranking of the form $(A_i, A_{i+1},....,A_n,A_1,\dots,A_{i-1})$ for $i=1,\dots,n$ where alternative $A_j$ is ranked higher than alternative $A_{j-1}$ is the ranking $(A_j, A_{j+1},....,A_n,A_1,\dots,A_{j-1})$. Hence at least $k$ of the $2k-1$ rankings $\sigma_1,\dots,\sigma_{2k-1}$ must be the ranking $(A_j, A_{j+1},....,A_n,A_1,\dots,A_{j-1})$.

We have shown that for $\sigma_1,\dots,\sigma_{2k-1}$ chosen according to the probability distribution $\pi^*$, any alternative $A_j$ is a Condorcet winner if and only if at least $k$ of the $2k-1$ rankings $\sigma_1,\dots,\sigma_{2k-1}$ are equal to $(A_j, A_{i+1},....,A_n,A_1,\dots,A_{j-1})$. For each $j=1,\dots,n$, we claim that the probability that this happens is precisely $p_k(1/n)$. Recall that each of the rankings $\sigma_1,\dots,\sigma_{2k-1}$ equals $(A_j, A_{j+1},....,A_n,A_1,\dots,A_{j-1})$ with probability $1/n$. Hence the probability that at least $k$ of the the $2k-1$ rankings $\sigma_1,\dots,\sigma_{2k-1}$ are $(A_j, A_{j+1},....,A_n,A_1,\dots,A_{j-1})$ is the same as the probability of having at least $k$ heads among $2k-1$ throws of a biased coin that shows heads with probability $1/n$. This latter probability is precisely $p_k(1/n)$.

So we have shown that for each $j=1,\dots,n$, alternative $A_j$ is a Condorcet winner with probability exactly $p_k(1/n)$. Since there is always at most one Condorcet winner, we can conclude that the probability of having a Condorcet winner equals $n\cdot p_k(1/n)$.
\end{proof}

\section{Proofs of technical lemmas}
\label{sect-lemmas}

It remains to prove Lemmas \ref{lem-integral-finite}, \ref{lemma-bounds-Q} and \ref{lem-taylor}, as well as Proposition \ref{propo-optimization}. We will prove the lemmas \ref{lemma-bounds-Q} and \ref{lem-taylor} in the first subsection, and we will give the (more complicated) proof of Lemma \ref{lem-integral-finite} in the second subsection. The proof of Proposition \ref{propo-optimization} can be found in the last subsection.

\subsection{Proof of Lemmas \ref{lemma-bounds-Q} and \ref{lem-taylor}}

\begin{proof}[Proof of Lemma \ref{lemma-bounds-Q}]
Recall that we defined $Q(x_1,\dots,x_{2k-1})$ to be the probability that independent uniformly random variables $y_1,\dots,y_{2k-1}\in [0,1]$ satisfy the condition that $x_i\leq y_i$ for at most $k-1$ indices $i\in [2k-1]$. Note that this condition is equivalent to saying that $y_i<x_i$ for at least $k$ indices $i\in [2k-1]$. Hence we can equivalently define $Q(x_1,\dots,x_{2k-1})$ to be the probability that independent uniformly random variables $y_1,\dots,y_{2k-1}\in [0,1]$ satisfy the condition that $y_i<x_i$ for at least $k$ indices $i\in [2k-1]$.

For a subset $I\su [2k-1]$ of size $|I|=k$, let us define $\mathcal{E}_I$ to be the event that we have $y_i<x_i$ for all $i\in I$. Then $Q(x_1,\dots,x_{2k-1})$ is the probability that at least one of the events $\mathcal{E}_I$ for some subset $I\su [2k-1]$ of size $|I|=k$ holds. In other words
\[Q(x_1,\dots,x_{2k-1})=\PP\left[\bigcup_{I}\mathcal{E}_I \right],\]
where the union is taken over all subsets $I\su [2k-1]$ of size $|I|=k$

For each such subset $I$ we have $\PP[\mathcal{E}_I]=\prod_{i\in I}x_i$, since for each $i\in I$ we have $y_i<x_i$ with probability $x_i$ and this happens independently for all $i\in I$. Thus, by the union bound we obtain
\[Q(x_1,\dots,x_{2k-1})=\PP\left[\bigcup_{I}\mathcal{E}_I \right]\leq \sum_{\substack{I\su [2k-1]\\ |I|=k}}\PP[\mathcal{E}_I]=\sum_{\substack{I\su [2k-1]\\ |I|=k}}\,\prod_{i\in I}x_i=\sigma_k(x_1,\dots,x_{2k-1}).\]
This proves the upper bound in the first part of Lemma \ref{lemma-bounds-Q}.

For the lower bound, note that $\sigma_k(x_1,\dots,x_{2k-1})=\sum_{I\su [2k-1],\, |I|=k}\prod_{i\in I}x_i$ is a sum of $\binom{2k-1}{k}\leq 2^{2k-1}$ summands. Hence at least one of these summands must be at least $2^{-2k+1}\cdot \sigma_k(x_1,\dots,x_{2k-1})$. In other words, there exists a subset $J\su [2k-1]$ of size $|J|=k$ such that $\prod_{i\in J}x_i\geq 2^{-2k+1}\cdot \sigma_k(x_1,\dots,x_{2k-1})$. Hence
\[Q(x_1,\dots,x_{2k-1})=\PP\left[\bigcup_{I}\mathcal{E}_I \right]\geq \PP[\mathcal{E}_J]=\prod_{i\in J}x_i\geq 2^{-2k+1}\cdot \sigma_k(x_1,\dots,x_{2k-1}).\]
This finishes the proof of the first part of the lemma.

For the second part of the lemma, we use that by Bonferroni's inequalities we have
\[Q(x_1,\dots,x_{2k-1})=\PP\left[\bigcup_{I}\mathcal{E}_I \right]\geq \sum_I \PP[\mathcal{E}_I]- \sum_{I, I'} \PP[\mathcal{E}_I\cap \mathcal{E}_{I'}]=\sigma_k(x_1,\dots,x_{2k-1})-\sum_{I, I'} \PP[\mathcal{E}_I\cap \mathcal{E}_{I'}],\]
where the last sum is over all choices of two distinct subsets $I,I'\su [2k-1]$ with $|I|=|I'|=k$. Note that for any choice of two such subsets, the event $\mathcal{E}_I\cap \mathcal{E}_{I'}$ happens if and only if $y_i<x_i$ for all $i\in I\cup I'$ and the probability for this to occur is precisely $\prod_{i\in I\cup I'}x_i$ (here, we again used that the different variables $y_i$ are independent). Thus,
\[Q(x_1,\dots,x_{2k-1})\geq \sigma_k(x_1,\dots,x_{2k-1})-\sum_{I, I'}\prod_{i\in I\cup I'}x_i,\]
where the sum is again over all choices of two distinct subsets $I,I'\su [2k-1]$ with $|I|=|I'|=k$. Note that for any two such subsets we have $|I\cup I'|\geq k+1$. Hence for any two such $I,I'$ we can choose a subset $J(I,I')\su [2k-1]$ of size $|J(I,I')|=k+1$ with $J(I,I')\su I\cup I'$. Then, as $x_i\in [0,1]$ for all $i$, we have
\[Q(x_1,\dots,x_{2k-1})\geq \sigma_k(x_1,\dots,x_{2k-1})-\sum_{I, I'}\prod_{i\in I\cup I'}x_i\geq \sigma_k(x_1,\dots,x_{2k-1})-\sum_{I, I'}\prod_{i\in J(I,I')}x_i.\]
Note that the sum on the right-hand side has less than $\binom{2k-1}{k}^2\leq 2^{4k-2}$ summands. Therefore each $J\su [2k-1]$ with $|J|=k+1$ occurs as $J(I,I')$ at most $2^{4k-2}$ times and we obtain
\[Q(x_1,\dots,x_{2k-1})\geq \sigma_k(x_1,\dots,x_{2k-1})-2^{4k-2}\cdot \sum_{\substack{J\su [2k-1]\\ |J|=k+1}}\,\prod_{i\in J}x_i=\sigma_k(x_1,\dots,x_{2k-1})-2^{4k-2}\cdot \sigma_{k+1}(x_1,\dots,x_{2k-1}).\]
Finally, by Maclaurin's inequality for elementary symmetric polynomials we have
\[\frac{\sigma_{k+1}(x_1,\dots,x_{2k-1})}{\binom{2k-1}{k+1}}\leq \left(\frac{\sigma_{k}(x_1,\dots,x_{2k-1})}{\binom{2k-1}{k}}\right)^{(k+1)/k}\]
and therefore, as $\binom{2k-1}{k+1}\leq \binom{2k-1}{k}$, we obtain $\sigma_{k+1}(x_1,\dots,x_{2k-1})\leq \sigma_{k}(x_1,\dots,x_{2k-1})^{(k+1)/k}$. Thus,
\[Q(x_1,\dots,x_{2k-1})\geq  \sigma_k(x_1,\dots,x_{2k-1})-2^{4k-2}\cdot\sigma_{k}(x_1,\dots,x_{2k-1})^{(k+1)/k},\]
as desired.
\end{proof}

\begin{proof}[Proof of Lemma \ref{lem-taylor}] For every $y<0$, by Taylor's theorem (with Lagrange remainder term) there is some $\xi$ in the interval $[y,0]$ such that
\[e^{y}=1+y+\frac{e^\xi}{2}\cdot y^2.\]
Using $0\leq e^\xi\leq 1$, we can conclude that
\[1+y\leq e^{y}\leq 1+y+\frac{1}{2}\cdot y^2\]
for all $y<0$. Now, let $0\leq t\leq 1/3$. Setting $y=-t$, we obtain $e^{-t}\geq 1+(-t)=1-t$, establishing the second inequality in Lemma \ref{lem-taylor}. For the first inequality, taking $y=-t-t^2$ gives
\[e^{-t-t^2}\leq 1-t-t^2+\frac{1}{2}\cdot (t+t^2)^2=1-t-\frac{1}{2}\cdot t^2+t^3+\frac{1}{2}\cdot t^4\le 1-t-\frac{1}{2}\cdot t^2+\frac{3}{2}\cdot t^3\leq 1-t,\]
as desired.
\end{proof}

\subsection{Proof of Lemma \ref{lem-integral-finite}}

We will use the following easy lemma in the proof of Lemma \ref{lem-integral-finite}.

\begin{lemma}\label{lemma-preparation-integral-finite}
For any positive integers $\ell<m$, and any $a\geq 1$, we have
\begin{multline*}
\int_0^\infty \dots \int_0^\infty \exp(-\sigma_\ell(x_1,\dots,x_{m}))\diff x_1 \dots \diff x_{m}\leq\\
\int_0^a \dots \int_0^a \exp(-\sigma_\ell(x_1,\dots,x_{m}))\diff x_1 \dots \diff x_{m}+m\cdot \int_a^\infty \int_0^\infty \dots \int_0^\infty \exp(-\sigma_\ell(x_1,\dots,x_{m}))\diff x_1 \dots \diff x_{m}
\end{multline*}
\end{lemma}
\begin{proof}
For any $a\geq 1$, and any $i=1,\dots, m$, let us define the domain
\[D_i^{(a)}=\{(x_1,\dots,x_{m})\in [0,\infty)^{m} \mid  x_i\geq a\}.\]
It is not hard to see that $[0,\infty)^{m}$ is covered by the union of $[0,a]^{m}$ and the sets $D_i^{(a)}$ for $i=1,\dots,m$. Thus, for any $a\geq 1$ we can conclude
\begin{align*}
&\int_0^\infty \dots \int_0^\infty \exp(-\sigma_\ell(x_1,\dots,x_{m}))\diff x_1 \dots \diff x_{m}\\
&\quad \leq \int_0^a \dots \int_0^a \exp(-\sigma_\ell(x_1,\dots,x_{m}))\diff x_1 \dots \diff x_{m}+\sum_{i=1}^{m} \int_{D_i^{(a)}} \exp(-\sigma_\ell(x_1,\dots,x_{m})) \diff^{m} (x_1,\dots,x_{m})\\
&\quad= \int_0^a \dots \int_0^a \exp(-\sigma_\ell(x_1,\dots,x_{m}))\diff x_1 \dots \diff x_{m}+m\cdot \int_{D_m^{(a)}} \exp(-\sigma_\ell(x_1,\dots,x_{m})) \diff^{m} (x_1,\dots,x_{m})\\
&\quad=\int_0^a \dots \int_0^a \exp(-\sigma_\ell(x_1,\dots,x_{m}))\diff x_1 \dots \diff x_{m}+m\cdot \int_a^\infty \int_0^\infty \dots \int_0^\infty \exp(-\sigma_\ell(x_1,\dots,x_{m}))\diff x_1 \dots \diff x_{m},
\end{align*}
where in the second step we used that by symmetry of the function $\exp(-\sigma_\ell(x_1,\dots,x_m))$ its integral has the same value on each of the domains $D_i^{(a)}$ for $i=1,\dots,m$.
\end{proof}

We can now prove Lemma \ref{lem-integral-finite} by induction on $\ell$.

\begin{proof}[Proof of Lemma \ref{lem-integral-finite}]
First, consider the case $\ell=1$. Then
\begin{multline*}
\int_0^\infty \dots \int_0^\infty \exp(-\sigma_1(x_1,\dots,x_{m}))\diff x_1 \dots \diff x_{m}=\int_0^\infty \dots \int_0^\infty \exp(-x_1-\dots-x_m)\diff x_1 \dots \diff x_{m}\\
=\left(\int_0^\infty e^{-x_1}\diff x_1\right)\dotsm \left(\int_0^\infty e^{-x_m}\diff x_1\right)=\left(\int_0^\infty e^{-x}\diff x\right)^m=1^m=1\leq  (m!)^2,
\end{multline*}
so Lemma \ref{lem-integral-finite} holds for $\ell=1$ (note that the second part of the lemma statement is only for the case of $\ell\geq 2$).

Now, let us assume that $\ell\geq 2$ and that we already proved Lemma \ref{lem-integral-finite} for $\ell-1$. We claim that in order to prove Lemma  \ref{lem-integral-finite} for $\ell$, it suffices to show that
\begin{equation}\label{eq-sufficient-induction-integral}
\int_a^\infty \int_0^\infty \dots \int_0^\infty \exp(-\sigma_\ell(x_1,\dots,x_{m}))\diff x_1 \dots \diff x_{m}\leq (m-1)\cdot ((m-1)!)^2 \cdot a^{-(m-\ell)/(\ell-1)}
\end{equation}
for any integer $m>\ell$ and any $a\geq 1$. Indeed, the inequality in the second part of Lemma  \ref{lem-integral-finite} follows directly by combining (\ref{eq-sufficient-induction-integral}) with Lemma \ref{lemma-preparation-integral-finite}. Furthermore, combining (\ref{eq-sufficient-induction-integral}) with Lemma \ref{lemma-preparation-integral-finite} in the special case of $a=1$ gives
\begin{align*}
&\int_0^\infty \dots \int_0^\infty \exp(-\sigma_\ell(x_1,\dots,x_{m}))\diff x_1 \dots \diff x_{m}\\
&\quad\leq \int_0^1 \dots \int_0^1 \exp(-\sigma_\ell(x_1,\dots,x_{m}))\diff x_1 \dots \diff x_{m}+m\cdot \int_1^\infty \int_0^\infty \dots \int_0^\infty \exp(-\sigma_\ell(x_1,\dots,x_{m}))\diff x_1 \dots \diff x_{m}\\
&\quad\leq \int_0^1 \dots \int_0^1 1\diff x_1 \dots \diff x_{m}+m\cdot (m-1)\cdot ((m-1)!)^2 \cdot 1^{-(m-\ell)/(\ell-1)}\\
&\quad =1+m\cdot (m-1)\cdot ((m-1)!)^2\leq m\cdot ((m-1)!)^2+m\cdot (m-1)\cdot ((m-1)!)^2= (m!)^2,
\end{align*}
which proves the first part of Lemma \ref{lem-integral-finite} for $\ell$. So it only remains to show (\ref{eq-sufficient-induction-integral}). Note that for any non-negative $x_1,\dots,x_m$ we have $\sigma_\ell(x_1,\dots,x_{m})\geq\sigma_{\ell-1}(x_1,\dots,x_{m-1})\cdot x_m$ (indeed, the right-hand side consists of precisely those terms of  $\sigma_\ell(x_1,\dots,x_{m})$ that contain $x_m$). Hence
\begin{align*}
&\int_a^\infty \int_0^\infty \dots \int_0^\infty \exp(-\sigma_\ell(x_1,\dots,x_{m}))\diff x_1 \dots \diff x_{m}\\
&\quad=\int_a^\infty \left(\int_0^\infty \dots \int_0^\infty \exp(-\sigma_\ell(x_1,\dots,x_{m}))\diff x_1 \dots \diff x_{m-1}\right)\diff x_{m}\\
&\quad\leq \int_a^\infty \left(\int_0^\infty \dots \int_0^\infty \exp(-\sigma_{\ell-1}(x_1,\dots,x_{m-1})\cdot x_m)\diff x_1 \dots \diff x_{m-1}\right)\diff x_{m}\\
&\quad= \int_a^\infty \left(x_m^{-(m-1)/(\ell-1)}\cdot \int_0^\infty \dots \int_0^\infty \exp(-\sigma_{\ell-1}(z_1,\dots,z_{m-1}))\diff z_1 \dots \diff z_{m-1}\right)\diff x_{m}\\
&\quad\leq \int_a^\infty x_m^{-\frac{m-1}{\ell-1}}\cdot ((m-1)!)^2\diff x_{m}=((m-1)!)^2\cdot \frac{\ell-1}{m-\ell}\cdot a^{-\frac{m-\ell}{\ell-1}} \leq (\ell-1)\cdot ((m-1)!)^2 \cdot a^{-\frac{m-\ell}{\ell-1}},
\end{align*}
where in the third step we considered the substitution $z_i=x_m^{1/(\ell-1)}x_i$ for $i=1,\dots,m-1$ (and used that $\sigma_{\ell-1}(x_1,\dots,x_{m-1})$ is a homogeneous polynomial of degree $\ell-1$), and in the fourth step we used the induction hypothesis for $\ell-1$ (noting that $m-1>\ell-1$). Since $\ell<m$, this in particular proves (\ref{eq-sufficient-induction-integral}).
\end{proof}

\subsection{Proof of Proposition \ref{propo-optimization}}
\label{subsect-propo-optimization}

We start by proving some lemmas about properties of the function $p_k(x)$, which will be used in our proof of Proposition \ref{propo-optimization}.

\begin{lemma}\label{lemma-function-anti-symmetry}
For every $k\geq 1$ and $x\in [0,1]$ we have $p_k(x)+p_k(1-x)=1$.
\end{lemma}
\begin{proof}
We have
\begin{align*}
p_k(x)+p_k(1-x)&=\sum_{\ell=0}^{k-1}\binom{2k-1}{\ell}\cdot x^{2k-1-\ell}\cdot (1-x)^{\ell}+\sum_{\ell=0}^{k-1}\binom{2k-1}{\ell}\cdot (1-x)^{2k-1-\ell}\cdot x^{\ell}\\
&=\sum_{\ell=0}^{k-1}\binom{2k-1}{\ell}\cdot x^{2k-1-\ell}\cdot (1-x)^{\ell}+\sum_{\ell=k}^{2k-1}\binom{2k-1}{2k-1-\ell}\cdot (1-x)^{\ell}\cdot x^{2k-1-\ell}\\
&=\sum_{\ell=0}^{2k-1}\binom{2k-1}{\ell}\cdot x^{2k-1-\ell}\cdot (1-x)^{\ell}\\
&=\left(x+(1-x)\right)^{2k-1}=1,
\end{align*}
where in the second-last step we used the binomial theorem.
\end{proof}

\begin{lemma}\label{lemma-function-derivative}
For $k\geq 1$ and $x\in [0,1]$, the derivative of the function $p_k(x)$ is
\[p_k'(x)=\frac{(2k-1)!}{(k-1)!^2}\cdot x^{k-1}(1-x)^{k-1}.\]
\end{lemma}
\begin{proof}
We have
\begin{align*}
p_k'(x)&=\sum_{\ell=1}^{k-1}\binom{2k-1}{\ell}\cdot \left((2k-1-\ell)x^{2k-2-\ell}(1-x)^{\ell}-\ell x^{2k-1-\ell}(1-x)^{\ell-1}\right)+\binom{2k-1}{0}\cdot (2k-1)x^{2k-2}\\
&=\sum_{\ell=0}^{k-1}\binom{2k-1}{\ell}\cdot (2k-1-\ell)x^{2k-2-\ell}(1-x)^{\ell}-\sum_{\ell=0}^{k-2}\binom{2k-1}{\ell+1}\cdot (\ell+1)x^{2k-2-\ell}(1-x)^{\ell}\\
&=\sum_{\ell=0}^{k-1}\left(\binom{2k-1}{\ell}(2k-1-\ell)-\binom{2k-1}{\ell+1}(\ell+1)\right)x^{2k-2-\ell}(1-x)^{\ell}+\binom{2k-1}{k-1}\cdot kx^{k-1}(1-x)^{k-1}\\
&=\frac{(2k-1)!}{(k-1)!^2}\cdot x^{k-1}(1-x)^{k-1},
\end{align*}
where in the second-last step we used that
\[\binom{2k-1}{\ell}(2k-1-\ell)-\binom{2k-1}{\ell+1}(\ell+1)=\frac{(2k-1)!}{\ell!\cdot (2k-2-\ell)!}-\frac{(2k-1)!}{\ell!\cdot (2k-2-\ell)!}=0\]
for all $\ell=0,\dots,k-1$.
\end{proof}

\begin{lemma}\label{lemma-function-convex}
For every $k\geq 1$, the function $p_k(x)$ is convex on the interval $[0,1/2]$.
\end{lemma}
\begin{proof}
Note that the function $x(1-x)=x-x^2=1/4-((1/2)-x)^2$ is non-negative and monotonically increasing on the interval $[0,1/2]$. Hence the function $x^{k-1}(1-x)^{k-1}$ is monotonically non-decreasing on $[0,1/2]$. By Lemma \ref{lemma-function-derivative}, this means that the derivative $p_k'(x)$ is monotonically non-decreasing on the interval $[0,1/2]$. Hence $p_k(x)$ is convex on this interval.
\end{proof}

\begin{lemma}\label{lemma-less-than-1}
For every $k\geq 1$ and every integer $n\geq 1$, we have $n\cdot p_k(1/n)\leq 1$.
\end{lemma}
\begin{proof}
For $n=1$, we have $1\cdot p_k(1/1)=p_k(1)= \binom{2k-1}{0}=1$, so the desired inequality is true.

For $n\geq 2$, note that $1/n\in [0,1/2]$. By Lemma \ref{lemma-function-convex} the function $p_k(x)$ is convex on the interval $[0,1/2]$ and by applying Jensen's inequality we obtain
\[p_k(1/n)=p_k\left(\frac{2}{n}\cdot \frac{1}{2}+\left(1-\frac{2}{n}\right)\cdot 0\right)\leq \frac{2}{n}\cdot p_k(1/2)+\left(1-\frac{2}{n}\right)\cdot p_k(0)=\frac{2}{n}\cdot \frac{1}{2}+\left(1-\frac{2}{n}\right)\cdot 0=\frac{1}{n}.\]
Here, we used that $p_k(1/2)=1/2$ by Lemma \ref{lemma-function-anti-symmetry} applied to $x=1/2$ and that $p_k(0)=0$. Hence $n\cdot p_k(1/n)\leq 1$, as desired.
\end{proof}

Let us now prove Proposition \ref{propo-optimization}.

\begin{proof}[Proof of Proposition \ref{propo-optimization}]
First, let us consider the case $k=1$. Then $p_k(x)=x$ for all $x\in [0,1]$, and so for any real numbers $x_1,\dots,x_n\in [0,1]$ with $x_1+\dots+x_n=1$ we clearly have $p_k(x_1)+\dots+p_k(x_n) =x_1+\dots+x_n=1=n\cdot p_k(1/n)$. So let us from now on assume that $k\geq 2$.

Recall that we need to show that $p_k(x_1)+\dots+p_k(x_n)\geq n\cdot p_k(1/n)$ for any real numbers $x_1,\dots,x_n\in [0,1]$ with $x_1+\dots+x_n=1$. We may thus assume that $x_1,\dots,x_n\in [0,1]$ are chosen to minimize $p_k(x_1)+\dots+p_k(x_n)$ under the constraint $x_1+\dots+x_n=1$.

First, consider the case that we have $x_1,\dots,x_n\in [0,1/2]$. Then, as the function $p_k(x)$ is convex on $[0,1/2]$ by Lemma \ref{lemma-function-convex}, Jensen's inequality implies 
\[p_k(x_1)+\dots+p_k(x_n)\geq n\cdot p_k((x_1+\dots+x_n)/n)=n\cdot p_k(1/n),\]
as desired.

Next, let us consider the case that we have $x_i=1$ for some $i\in \{1,\dots,n\}$. Then the remaining variables $x_1,\dots,x_{i-1},x_{i+1},\dots,x_n$ must all be zero, and we have
\[p_k(x_1)+\dots+p_k(x_n)=p_k(1)+(n-1)\cdot p_k(0)=1+(n-1)\cdot 0=1\geq n\cdot p_k(1/n),\]
where the last inequality is by Lemma \ref{lemma-less-than-1}.

So it remains to consider the case that we have $x_i\in (1/2,1)$ for some $i\in \{1,\dots,n\}$. Without loss of generality, let us assume that $x_n\in (1/2,1)$. Now, since $x_1+\dots +x_n=1$, at least one of  $x_1,\dots,x_{n-1}$ must be positive. Again without loss of generality let us assume that $x_1>0$, then $0<x_1<x_1+x_n\leq 1$. Now, for $t\in [0,x_1+x_n]$ consider the function $g(t)=p_k(t)+p_k(x_2)+\dots+p_k(x_{n-1})+p_k(x_1+x_n-t)$. As $t+x_2+\dots+x_{n-1}+(x_1+x_n-t)=x_1+\dots+x_n=1$, by the choice of $x_1,\dots,x_n$ to minimize $p_k(x_1)+\dots+p_k(x_n)$, we must have
\[g(t)=p_k(t)+p_k(x_2)+\dots+p_k(x_{n-1})+p_k(x_1+x_n-t)\geq p_k(x_1)+\dots+p_k(x_n)=g(x_1)\]
for all $t\in [0,x_1+x_n]$. In other words, the function $g(t)$ has a minimum at $t=x_1$. Hence, using that $0<x_1<x_1+x_n$, the derivative $g'(t)=p_k'(t)-p_k'(x_1+x_n-t)$ of  $g(t)$ at $t=x_1$ must satisfy $g'(x_1)=0$. From Lemma \ref{lemma-function-derivative}, we obtain
\[0=g'(x_1)=p_k'(x_1)-p_k'(x_1+x_n-x_1)=p_k'(x_1)-p_k'(x_n)=\frac{(2k-1)!}{(k-1)!^2}\cdot \left(x_1^{k-1}(1-x_1)^{k-1}-x_n^{k-1}(1-x_n)^{k-1}\right).\]
Hence $x_1^{k-1}(1-x_1)^{k-1}=x_n^{k-1}(1-x_n)^{k-1}$ and, as $k\geq 2$, this implies that $x_1(1-x_1)=x_n(1-x_n)$. Therefore
\[(x_1-x_n)\cdot (x_1+x_n-1)=x_n(1-x_n)-x_1(1-x_1)=0.\]
and consequently we must have $x_1=x_n$ or $x_1+x_n=1$. However, as $x_n\in (1/2,1)$, the former equation is impossible (as otherwise $x_1+\dots+x_n\geq x_1+x_n=2x_n>1$). Thus, we can conclude that $x_1+x_n=1$ and therefore $x_2=\dots=x_{n-1}=0$. Now,
\[p_k(x_1)+\dots+p_k(x_n)=p_k(x_1)+p_k(x_n)+(n-2)\cdot p_k(0)=p_k(x_1)+p_k(1-x_1)+(n-2)\cdot 0=1\geq n\cdot p_k(1/n),\]
where in the second-last step we used Lemma \ref{lemma-function-anti-symmetry} and in the last step Lemma \ref{lemma-less-than-1}. This finishes the proof of Proposition \ref{propo-optimization}.
\end{proof}

\end{document}